\documentclass[draftclsnofoot,onecolumn,12pt]{IEEEtran}

\usepackage{amsmath,amsthm,amsfonts,amssymb,dsfont,mathtools,bm}
\usepackage{graphicx}
\usepackage{xcolor}
\usepackage{etoolbox}
\newif\ifmarkrevisions
\markrevisionsfalse
\newcommand{\rev}[1]{\ifmarkrevisions\textcolor{blue}{#1}\else#1\fi}
\newenvironment{revision}{\begingroup\ifmarkrevisions\color{blue}\fi}{\par\endgroup}
\newcommand{\markrevisedbib}[1]{\csdef{revisedbib@#1}{}}
\forcsvlist{\markrevisedbib}{GrafLuschgy2000,LiuPages2020,RockafellarWets1998}
\newcommand{\revisionbibcolor}[1]{\ifmarkrevisions\ifcsdef{revisedbib@#1}{\color{blue}}{\normalcolor}\fi}
\makeatletter
\pretocmd{\@bibitem}{\revisionbibcolor{#1}}{}{}
\pretocmd{\@lbibitem}{\revisionbibcolor{#2}}{}{}
\makeatother
\usepackage{relsize}
\usepackage{algorithmicx}
\usepackage{algorithm}
\usepackage{float}
\usepackage{placeins}
\usepackage{cleveref}
\usepackage{cite}
\usepackage{cancel}

\graphicspath{{figures/}}

\theoremstyle{plain}
\newtheorem{lem}{Lemma}
\newtheorem{theorem}{Theorem}
\newtheorem{proposition}{Proposition}
\newtheorem{coro}{Corollary}

\theoremstyle{definition}

\theoremstyle{remark}
\newtheorem{remark}{Remark}

\newcommand{\e}{{\rm e}}
\newcommand{\D}{\mathrm{d}}

\newcommand\bigexists{%
	\mathop{\lower0.75ex\hbox{\ensuremath{%
				\mathlarger{\mathlarger{\exists}}}}}%
	\limits}

\title{List Estimation}

\author{Nikola Zlatanov, Amin  Gohari, Farzad Shahrivari, and Mikhail Rudakov
\thanks{N. Zlatanov and M. Rudakov are with the Faculty of Computer Science and Software Engineering, Innopolis University, Innopolis, Russia. E-mail: n.zlatanov@innopolis.ru}
\thanks{ A. Gohari is with the Department of Information Engineering, The Chinese University of Hong Kong, Hong Kong, China.}
\thanks{F. Shahrivari was with the Department of Electrical and Computer Systems Engineering, Monash University, Melbourne, Australia.  }
 }

\begin{document}

\maketitle

\begin{abstract}
Classical estimation outputs a single point estimate of an unknown $d$-dimensional vector from an observation.
 In this paper, we study \emph{$k$-list estimation}, in which a single observation is used to produce a list of $k$ candidate estimates and performance is measured by the expected squared distance from the true vector to the closest candidate. We compare this centralized setting with a symmetric decentralized MMSE benchmark in which $k$ agents observe conditionally i.i.d.\ measurements and each agent outputs its own MMSE estimate. On the centralized side, we show that optimal $k$-list estimation is equivalent to fixed-rate $k$-point vector quantization of the posterior distribution and, under standard regularity conditions, admits an exact high-rate asymptotic expansion with explicit constants and decay rate $k^{-2/d}$. On the decentralized side, we derive lower bounds in terms of the small-ball behavior of the single-agent MMSE error; in particular, when the conditional error density is bounded near the origin, the benchmark distortion cannot decay faster than order $k^{-2/d}$. We further show that if the error density vanishes at the origin
in a power-law fashion, then the decentralized benchmark is provably unable to match the centralized $k^{-2/d}$ exponent, whereas the centralized estimator retains the $k^{-2/d}$ scaling. Gaussian specializations yield explicit formulas  and numerical experiments corroborate the predicted asymptotic behavior. Overall, the results show that, in the scaling with $k$, one observation combined with $k$ carefully chosen candidates can be asymptotically as effective as---and in some regimes strictly better than---this MMSE-based decentralized benchmark with $k$ \rev{conditionally independent} observations.
\end{abstract}


\section{Introduction}\label{Sec1}

\IEEEPARstart{C}{lassical} estimation takes one observation and returns one estimate. Under squared-error loss, the canonical choice is the minimum mean-squared error (MMSE) estimator, given by the conditional expectation $\mathbb{E}[\bm X|\bm Y]$ \cite{RePEc:eee:monogr:9780122951831}. In many systems, however, a single estimate is not the most useful output. A downstream stage may be able to test, refine, or select among several plausible candidates generated from the same measurement. Motivated by this setting, we study \emph{$k$-list estimation}: from a single observation $\bm Y$, an estimator produces a list of $k$ candidates $\widehat{\bm X}_1,\ldots,\widehat{\bm X}_k$, and performance is judged by how close the nearest candidate is to the unknown target vector $\bm X\in\mathbb{R}^d$.

This formulation leads to a natural question: how does one observation plus $k$ carefully designed candidates compare with $k$ \rev{conditionally independent} observations processed separately by $k$ agents? To make the comparison concrete, we study a centralized list estimator and a symmetric decentralized MMSE benchmark in which agent $i$ observes $\bm Y_i$ and outputs $\mathbb{E}[\bm X|\bm Y_i]$. This benchmark is natural and practically important\footnote{We do not claim this benchmark to be the globally optimal decentralized architecture for the min-of-$k$ objective.}. The comparison is deliberately conservative with respect to the centralized side: the list estimator receives only a \emph{single} observation, whereas the decentralized benchmark has access to $k$ \rev{conditionally independent} observations.

List estimation is relevant whenever acquiring a new observation is expensive, while generating and testing multiple candidates from an existing observation is comparatively cheap. Examples include multi-robot target search, where detection time is dominated by the robot initialized closest to the target; parallel experimental design and hyperparameter search, where success is achieved if at least one configuration is close to optimal; beam sweeping in millimeter-wave systems, where performance depends on the beam whose direction is closest to the channel direction; and resource-constrained medical imaging, where the same low-cost measurements may be used to generate several plausible reconstructions before deciding whether higher-fidelity follow-up is needed. In all of these examples, system performance is governed by the \emph{best} candidate among several parallel options.

The main message of the paper is that, under broad regularity conditions, the best-candidate distortion from one observation scales surprisingly well with the list size. In particular, the centralized list estimator achieves   $k^{-2/d}$ high-rate scaling, while the decentralized MMSE benchmark cannot improve upon that exponent in smooth models and is provably unable to match it in regimes where very small MMSE errors are sufficiently rare.

\rev{The paper develops an estimation-theoretic framework that quantifies the
advantage of jointly designing $k$ candidates from one observation relative
to forming estimates separately from $k$ observations. We characterize
optimal list estimation through posterior vector quantization, establish
its exact high-rate asymptotics, and derive coordinate-free decentralized
lower bounds through small-ball probabilities. These results identify
conditions under which even optimally chosen non-identical deterministic
or randomized local rules of the decentralized benchmark cannot improve upon the  exponent
$2/d$ of the centralized list estimation, together with informative regimes in which the decentralized MMSE benchmark is provably unable to attain that exponent.}

\subsection{Choice of performance metric}\label{subsec:metric}

We study the distortion
\begin{equation*}
D(k)=\mathbb{E}\!\left[\,\min_{1\le i\le k}\big\|\bm X-\widehat{\bm X}_i\big\|_2^2\,\right].
\end{equation*}
This metric is appropriate when a downstream stage can exploit the list, for example by refining several candidates in parallel, testing them locally, or using a short feedback signal to identify the most promising one. In that sense, the list is an intermediate representation rather than necessarily the final decision.

This viewpoint is analogous to list decoding, where an inner decoder outputs several candidates and a later stage resolves the ambiguity. We do not model that later stage explicitly. Instead, we use the minimum-over-$k$ squared error as an analytically tractable proxy for systems that can capitalize on the best candidate. When no such downstream mechanism exists, this metric is more optimistic than conventional single-estimate MSE, and our claims should be interpreted accordingly.

\subsection{Contributions}

Our contributions are as follows.

\noindent\textbf{(1) Centralized formulation and high-rate asymptotics:}
We formalize $k$-list estimation from a single observation and show that, for each realized observation, the optimal list is the optimal $k$-point vector quantizer of the posterior distribution. Under \rev{the stated smoothness and averaged moment conditions}, the resulting distortion admits an exact high-rate asymptotic expansion of order $k^{-2/d}$ with explicit constants.

\noindent\textbf{(2) Lower bounds for a decentralized MMSE benchmark:}
We study a symmetric benchmark in which $k$ agents observe conditionally i.i.d.\ measurements and each agent applies the single-estimate MMSE rule. For this benchmark we derive lower bounds in terms of the small-ball probability of the single-agent MMSE error near the origin. The bounds require only conditional i.i.d.\ observations across agents and do not rely on coordinate-wise independence assumptions. \rev{Under an integrable bound on the posterior density, we also obtain a lower bound uniform over non-identical deterministic and randomized local rules.}

\noindent\textbf{(3) Exponent comparison via local error geometry:}
We compare the scaling of the centralized and decentralized distortions according to the local behavior of the conditional joint density of the MMSE error vector near the origin. If the conditional error density is bounded near zero, the decentralized benchmark cannot decay faster than $k^{-2/d}$, \rev{so it cannot improve upon the centralized exponent.} If the density vanishes at the origin in a power-law fashion, the benchmark is provably unable to match the centralized exponent, whereas the centralized estimator retains the $k^{-2/d}$ scaling.

\noindent\textbf{(4) Gaussian specialization and numerical validation:}
For Gaussian models we derive explicit centralized high-rate formulas,  verify the smooth-error regime for the decentralized benchmark, and use simulations to confirm the predicted asymptotic behavior.

Taken together, these results show that one observation plus a well-designed list of $k$ candidates \rev{achieves a $k$-exponent at least as good as that of} the MMSE-based decentralized benchmark with $k$ observations, and can outperform it in some regimes.

\subsection{Related works}

The centralized part of our analysis is closest in
technique to fixed-rate vector quantization and high-rate
distortion theory \cite{Abut:1990:VQ:533930,116036,720541},
with broader connections to rate-distortion and lossy
compression
\cite{Cover:2006:EIT:1146355,berger2003rate,berger1975rate,
blahut1972computation,ahmed1974discrete}.
The vector-quantization viewpoint underpins our centralized
analysis.

Conceptually related list-based reconstruction ideas appear in information-theoretic security. The authors of \cite{7454714} study an eavesdropper that produces a list of reconstruction sequences and measure secrecy through the minimum distortion over the list. Related distortion-based secrecy formulations for Shannon
ciphers include the exiguous-distortion probability criterion
in \cite{7805336} and the list-reconstruction criterion
in \cite{7855830}. These works focus on secrecy metrics and coding constructions, whereas our focus is estimation-theoretic scaling and comparison with decentralized MMSE estimation.

On the application side, \cite{PhysRevE.106.024101}
investigates target search with multiple independent random
walkers and optimizes their initial distribution using a
given probability distribution of the target location. Our setting differs in both objective and analysis: we study estimation under a min-of-$k$ squared-error criterion and derive asymptotic scaling laws. Another conceptual relative is the information-theoretic guessing literature \cite{Massey1994Guessing,Arikan1996Guessing}, where one sequentially guesses the value of a random variable---possibly with side information---until success. Guessing includes explicit feedback after each trial, unlike our setting, but both problems concern multiple candidate values of an unknown quantity.

\subsection{Notation}

Vectors are denoted by boldface letters; random variables are denoted by capital letters, and realizations by lowercase letters. Thus, boldface capital letters denote random vectors. We use $\mathbb{P}[\cdot]$ and $\mathbb{E}[\cdot]$ for probability and expectation, respectively. The PDFs of $\bm X$, $\bm X|\bm Y$, and $(\bm X,\bm Y)$ are denoted by $f_{\bm X}(\bm x)$, $f_{\bm X|\bm Y}(\bm x|\bm y)$, and $f_{\bm X,\bm Y}(\bm x,\bm y)$, respectively. We write $X\sim\mathcal{N}(\mu,\sigma^2)$ for a Gaussian random variable with mean $\mu$ and variance $\sigma^2$, and $\|\cdot\|_r$ for the $\ell_r$-norm.

We use the standard asymptotic notations $O(\cdot)$, $o(\cdot)$, $\Theta(\cdot)$, and $\Omega(\cdot)$ with respect to $k\to\infty$ unless stated otherwise. In particular, $f(k)=\Theta(g(k))$ means that $f(k)$ is bounded above and below by positive constant multiples of $g(k)$ for all sufficiently large $k$, while $f(k)=\Omega(g(k))$ denotes a corresponding lower bound. The parameters $d$ and $k$ are reserved for the ambient dimension and the list size, respectively. We denote the volume of the $d$-dimensional Euclidean unit ball by
\begin{equation}\label{eq:Vd_def_notation}
V_d \triangleq \mathrm{Vol}\big(\{\bm z:\|\bm z\|_2\le 1\}\big)=\frac{\pi^{d/2}}{\Gamma\!\left(\frac d2+1\right)},
\end{equation}
and the corresponding surface area by $S_d=dV_d$. Finally, $[1:k]$ denotes the set of integers $\{1,2,\ldots,k\}$.

\subsection{Paper structure}

The rest of the paper is organized as follows. Section~\ref{sec:Modeling} formulates the centralized $k$-list estimation problem and the decentralized MMSE benchmark. Section~\ref{sec:main_results} presents the main asymptotic results in general dimension $d$ and compares the corresponding distortion exponents. \rev{Section~\ref{sec:gaussian_examples} specializes the analysis to Gaussian models, computes explicit constants, and gives a non-Gaussian vanishing-density example.} Section~\ref{sec:simres} reports numerical experiments. Section~\ref{sec_conc} concludes the paper.


\section{Problem Formulation and Distortion Definitions}\label{sec:Modeling}

We now formalize the two estimation architectures compared throughout the paper. Let $\bm X\in\mathbb{R}^d$ denote the unknown vector of interest, and let $\bm Y\in\mathbb{R}^m$ be an observation generated according to a conditional PDF $f_{\bm Y|\bm X}(\bm y|\bm x)$, where the observation dimension $m$ need not equal $d$. \rev{We assume that $\mathbb E\|\bm X\|_2^2<\infty$.} The centralized architecture uses a \emph{single} observation $\bm Y$ to jointly produce a list of $k$ candidate estimates. The decentralized benchmark uses $k$ conditionally i.i.d.\ observations $\bm Y_1,\ldots,\bm Y_k$, each processed separately by one agent. In both cases, performance is measured by the \emph{best-candidate distortion}
\[
\mathbb{E}\!\left[\min_{1\le i\le k}\big\|\bm X-\widehat{\bm X}_i\big\|_2^2\right].
\]

\subsection{Centralized $k$-List Estimation}

A centralized $k$-list estimator consists of measurable mappings
$g_1,\ldots,g_k:\mathbb{R}^m\to\mathbb{R}^d$.
From one realization of $\bm Y$, it outputs the jointly designed list
\[
\widehat{\bm X}_i = g_i(\bm Y),\qquad i\in[1:k].
\]
The corresponding optimal best-candidate distortion is
\begin{equation}
\label{eq:D1_def}
D_1(k)\triangleq
\inf_{g_1(\cdot),\ldots,g_k(\cdot)}
\mathbb{E}\!\left[\min_{1\le i\le k}\big\|\bm X-g_i(\bm Y)\big\|_2^2\right].
\end{equation}

Conditioned on a realized observation $\bm Y=\bm y$, the design problem reduces to choosing $k$ deterministic representatives in $\mathbb{R}^d$ for the posterior distribution of $\bm X$. This yields the pointwise representation of the distortion as
\begin{align}
D_1(k)
&=\mathbb{E}_{\bm Y}\Bigg[
\inf_{\widehat{\bm x}_1,\ldots,\widehat{\bm x}_k\in\mathbb{R}^d}
\mathbb{E}\!\left[\min_{1\le i\le k}\big\|\bm X-\widehat{\bm x}_i\big\|_2^2\,\Big|\,\bm Y\right]
\Bigg]\label{eq:D1_pointwise}\\
&=\mathbb{E}_{\bm Y}\Bigg[
\inf_{\widehat{\bm x}_1,\ldots,\widehat{\bm x}_k\in\mathbb{R}^d}
\int_{\mathbb{R}^d}\min_{1\le i\le k}\big\|\bm x-\widehat{\bm x}_i\big\|_2^2\,
f_{\bm X|\bm Y}(\bm x|\bm Y)\,\D\bm x
\Bigg],\label{eq:D1_pointwise_integral}
\end{align}
where \eqref{eq:D1_pointwise_integral} assumes that the conditional distribution of $\bm X|\bm Y$ admits a density.

\begin{revision}
The interchange of infimum and expectation in \eqref{eq:D1_pointwise} follows from \cite[Theorem~14.60]{RockafellarWets1998}. The conditional distortion is measurable in $\bm y$ and continuous in the $k$ estimation points for almost every $\bm y$, since the conditional second moment is finite. Thus it is a normal integrand. The admissible list mappings comprise all measurable mappings into $(\mathbb R^d)^k$, a decomposable class, and the constant zero list has finite expected distortion. These properties give the stated equality of infima.
\end{revision}

For each realized $\bm y$, the inner optimization in \eqref{eq:D1_pointwise_integral} is exactly the optimal $k$-point vector quantization problem for the posterior density $f_{\bm X|\bm Y}(\cdot|\bm y)$ under squared-error distortion. Hence $D_1(k)$ is the average optimal posterior quantization error.

\subsection{Decentralized MMSE Benchmark}

In the decentralized benchmark, $k$ agents observe $\bm Y_1,\ldots,\bm Y_k\in\mathbb{R}^m$, where agent $i$ observes only $\bm Y_i$, $i\in[1:k]$. Given $\bm X$, these observations are assumed to be conditionally i.i.d.\ with the same kernel $f_{\bm Y|\bm X}$. Equivalently,
\begin{equation}\label{eq:agent_cond_iid_factorization}
f_{\bm X,\bm Y_1,\ldots,\bm Y_k}(\bm x,\bm y_1,\ldots,\bm y_k)
=
f_{\bm X}(\bm x)\prod_{i=1}^k f_{\bm Y|\bm X}(\bm y_i|\bm x).
\end{equation}

Let
\begin{equation}\label{eq:g}
g(\bm y)\triangleq \mathbb{E}\big[\bm X\big|\bm Y=\bm y\big]
\end{equation}
denote the single-agent MMSE estimator. Each agent applies the same rule to its own observation and outputs
\[
\widehat{\bm X}_i=g(\bm Y_i),\qquad i\in[1:k].
\]
The resulting best-candidate distortion is
\begin{equation}\label{eq:D2_def}
D_2(k)=\mathbb{E}\bigg[\min_{1\leq i\leq k}\Big\lVert\bm X-g(\bm Y_i)\Big\rVert_2^2\bigg].
\end{equation}
Throughout the paper, $D_2(k)$ denotes the distortion of this specific symmetric MMSE benchmark.  

\subsection{Modeling remarks and scope}\label{subsec:decentralized-benchmark}

The benchmark in \eqref{eq:D2_def} is natural because MMSE estimation is the canonical squared-error rule: if agents are designed independently to minimize their own mean-squared error, they all use $g$ given by \eqref{eq:g}. At the same time, once performance is evaluated through a minimum over agents, diversity of outputs can matter, so the MMSE rule need not be optimal from a global min-of-$k$ perspective.

\begin{remark}[Why MMSE may be suboptimal for the min-of-$k$ objective]
The mapping that minimizes $\mathbb{E}\|\bm X-g(\bm Y)\|_2^2$ for a single agent, i.e., when $k=1$, need not minimize
\[
\mathbb{E}\!\left[\min_{1\le i\le k}\|\bm X-g(\bm Y_i)\|_2^2\right]
\]
when there are multiple agents, i.e., when $k>1$. Intuitively, two agents with slightly worse individual MSE but more dispersed outputs may reduce the probability that \emph{all} agents are far away from the target. \rev{Appendix~\ref{ex:non_mmse_rule}  gives an exact two-agent Gaussian comparison demonstrating this effect.} \rev{However, under the posterior-density condition of Proposition~\ref{prop:general_local_bound}, even the optimal decentralized
design cannot have a distortion decay exponent exceeding $2/d$.
Characterizing the optimal mappings and the exact distortion remains
beyond the scope of this paper.}
\end{remark}

\begin{remark}[Conservative comparison]
The centralized estimator in \eqref{eq:D1_def} uses a \emph{single} observation, whereas the decentralized benchmark in \eqref{eq:D2_def} uses $k$ \rev{conditionally independent} observations. The comparison therefore favors the decentralized side in terms of available data. If the centralized estimator were also allowed access to $k$ observations, its performance could only improve.
\end{remark}

\begin{remark}[Degenerate side information]
If $\bm Y$ is independent of $\bm X$, the centralized $k$-list estimator can still place its $k$ candidates to quantize the prior distribution of $\bm X$. In contrast, the decentralized MMSE benchmark collapses to the identical outputs
\[
\widehat{\bm X}_i=\mathbb{E}[\bm X],\qquad i\in[1:k],
\]
so the min-of-$k$ criterion yields no gain from having multiple agents. \rev{This is a limiting case of the specified MMSE benchmark in which the observation $\bm Y$ is uninformative about the target $\bm X$. Our analysis addresses both uninformative observations, for which
$\bm Y$ and $\bm X$ are independent, and informative observations,
for which they are dependent.}
\end{remark}

\begin{remark}[Sanity check: $k=1$]
When $k=1$, the centralized problem in \eqref{eq:D1_def} reduces to standard MMSE estimation from one observation, and the benchmark in \eqref{eq:D2_def} reduces to the same single-agent MMSE estimator. Hence
$
D_1(1)=D_2(1).
$
\end{remark}


\section{Main Results in General Dimension}\label{sec:main_results}

This section establishes the main asymptotic comparison in general dimension $d$. We first derive an exact high-rate expression for the centralized best-candidate distortion $D_1(k)$. We then derive a converse-style lower bound for the decentralized MMSE benchmark $D_2(k)$ in terms of the small-ball behavior of the single-agent MMSE error. \rev{We then compare the resulting decay exponents in representative local-density regimes near the origin and extend the lower bounds of the decentralized benchmark to other local estimators.}

\subsection{Centralized $k$-List Estimation}

Once the formulation in Section~\ref{sec:Modeling} is recognized as posterior vector quantization, the main question is how the optimal posterior quantization error scales with the list size $k$. To streamline the notation, define the posterior Zador functional
\begin{equation}\label{eq:posterior_zador_functional}
\mathcal{J}(\bm y)\triangleq
\left(\int_{\mathbb{R}^d} f_{\bm X|\bm Y}(\bm x|\bm y)^{\frac{d}{d+2}}\,\D\bm x\right)^{\!\frac{d+2}{d}}.
\end{equation}

The following theorem gives the exact high-rate asymptotics of the centralized $k$-list estimator. The assumptions are standard sufficient conditions from high-rate quantization theory and are not intended to be minimal.

\begin{theorem}\label{thm_2a}
Assume the following conditions hold:
\begin{enumerate}
\item For almost every $\bm y$, the conditional distribution of $\bm X|\bm Y=\bm y$ admits a density $f_{\bm X|\bm Y}(\bm x|\bm y)$ with respect to Lebesgue measure on $\mathbb{R}^d$.
\item For almost every $\bm y$, the mapping $\bm x\mapsto f_{\bm X|\bm Y}(\bm x|\bm y)$ is bounded and differentiable in $\bm x$, and there exists $\delta>0$ such that
\[
\int_{\mathbb{R}^d}\|\bm x\|_2^{2+\delta}\,f_{\bm X|\bm Y}(\bm x|\bm y)\,\D\bm x < \infty.
\]
\begin{revision}
For the same $\delta$, assume also that
\[
\mathbb E_{\bm Y}\!\left[
\left(\int_{\mathbb R^d}\|\bm x\|_2^{2+\delta}
f_{\bm X|\bm Y}(\bm x|\bm Y)\,\D\bm x\right)^{\!2/(2+\delta)}
\right]<\infty.
\]
\end{revision}
\item The posterior Zador functional is integrable:
\[
\mathbb{E}_{\bm Y}\!\big[\mathcal{J}(\bm Y)\big] < \infty.
\]
\end{enumerate}
Then the best-candidate distortion of the centralized $k$-list estimator satisfies
\begin{equation}\label{eq:general_constant}
D_1(k)= G_d\,k^{-2/d}\,\mathbb{E}_{\bm Y}\!\big[\mathcal{J}(\bm Y)\big] + o\!\left(k^{-2/d}\right),
\end{equation}
where $G_d$ is the squared-error Zador constant in
$\mathbb R^d$, with distortion measured by
$\|\cdot\|_2^2$ \cite{1056490,720541,GrafLuschgy2000}.
In particular, if $\mathbb{E}[\mathcal{J}(\bm Y)]\in(0,\infty)$, then
\[
D_1(k)=\Theta\!\big(k^{-2/d}\big).
\]
\end{theorem}
\begin{revision}
\begin{IEEEproof}
Let $D_1(k\mid\bm y)$ denote the inner optimal posterior distortion in \eqref{eq:D1_pointwise_integral}. Zador's high-rate theorem \cite[Theorem~6.2]{GrafLuschgy2000} gives, for almost every $\bm y$,
\[
\lim_{k\to\infty}k^{2/d}D_1(k\mid\bm y)=G_d\mathcal J(\bm y).
\]
Moreover, the standard non-asymptotic vector quantization bound \cite[Eq.~(10)]{LiuPages2020}, squared to express mean-squared distortion, yields
\[
0\le k^{2/d}D_1(k\mid\bm y)
\le C_{d,\delta}
\left(\int_{\mathbb R^d}\|\bm x\|_2^{2+\delta}
f_{\bm X|\bm Y}(\bm x|\bm y)\,\D\bm x\right)^{\!2/(2+\delta)},
\]
where $C_{d,\delta}$ is independent of $k$ and $\bm y$. The right-hand side is integrable by Assumption~2. Dominated convergence and \eqref{eq:D1_pointwise_integral} therefore give
\[
\lim_{k\to\infty}k^{2/d}D_1(k)
=\mathbb E_{\bm Y}\!\left[\lim_{k\to\infty}k^{2/d}D_1(k\mid\bm Y)\right]
=G_d\mathbb E_{\bm Y}[\mathcal J(\bm Y)],
\]
which is equivalent to \eqref{eq:general_constant}.
\end{IEEEproof}
\end{revision}

\begin{coro}[Scalar specialization of Theorem~\ref{thm_2a} ($d=1$)]
\label{coro:D1_d1}
Assume that for almost every $y$, the conditional density $f_{X|Y}(x|y)$ is bounded and differentiable with respect to $x$, and there exists $\delta>0$ such that
\[
\int_{\mathbb{R}} |x|^{2+\delta}\,f_{X|Y}(x|y)\,\D x < \infty.
\]
\begin{revision}
For the same $\delta$, assume that
\[
\mathbb E_Y\!\left[
\left(\int_{\mathbb R}|x|^{2+\delta}f_{X|Y}(x|Y)\,\D x\right)^{\!2/(2+\delta)}
\right]<\infty.
\]
\end{revision}
Assume also that
\[
\int_{\mathbb{R}} f_{X|Y}(x|y)^{1/3}\,\D x < \infty,
\qquad
\mathbb{E}_{Y}[c(Y)]<\infty,
\]
where
\begin{equation}
    c(y) = \dfrac{1}{12}\left(\int_{\mathbb{R}} f_{X|Y}(x|y)^{1/3}\,\D x\right)^{3}.
\label{eq:c_of_y}
\end{equation}
Then
\begin{equation}\label{eq:D1_d1}
    D_1(k)=\dfrac{1}{k^2}\,\mathbb{E}_{Y}[c(Y)] + o\!\left(\dfrac{1}{k^2}\right).
\end{equation}
\end{coro}

\subsection{Decentralized MMSE Benchmark}

Let
\[
g(\bm y)\triangleq \mathbb{E}[\bm X|\bm Y=\bm y]
\]
denote the single-agent MMSE estimator, and define the single-agent MMSE error and squared error by
\begin{equation}\label{eq:single_agent_error}
\bm Z \triangleq \bm X-g(\bm Y),
\qquad
W \triangleq \|\bm Z\|_2^2.
\end{equation}
For the decentralized MMSE benchmark with conditionally i.i.d.\ observations $\bm Y_1,\ldots,\bm Y_k$, define
\begin{equation}\label{eq:Zi_def_new}
\bm Z_i \triangleq \bm X-g(\bm Y_i),
\qquad
W_i \triangleq \|\bm Z_i\|_2^2,
\qquad
\Lambda_k \triangleq \min_{1\le i\le k} W_i.
\end{equation}
Then $D_2(k)=\mathbb{E}[\Lambda_k]$.

The decentralized lower bound is governed by the probability that a \emph{single} MMSE estimate is exceptionally accurate.

\medskip
\noindent\textbf{Averaged small-ball condition:}
We assume that there exist constants $C>0$, $\alpha>0$, and $a_0>0$ such that
\begin{equation}\label{eq:ball_assumption}
\mathbb{P}\!\left(W\le a\right)\le C a^\alpha,
\qquad a\in[0,a_0].
\end{equation}
We refer to $\alpha$ as a \emph{small-ball exponent}. A simple sufficient condition for \eqref{eq:ball_assumption} is the pointwise conditional bound
\begin{equation}\label{eq:small_ball_assumption}
\mathbb{P}\!\left(W\le a \,\middle|\, \bm X=\bm x\right)\le C a^\alpha,
\qquad \forall\,\bm x\in\mathbb{R}^d,\ \forall\,a\in[0,a_0],
\end{equation}
since averaging over $\bm X$ yields \eqref{eq:ball_assumption}.

\begin{theorem}
\label{th_2}
Assume that the conditional-i.i.d.\ factorization \eqref{eq:agent_cond_iid_factorization} holds and that the averaged small-ball condition \eqref{eq:ball_assumption} is satisfied. Then, for every $k$ such that
\[
a^\star \triangleq \left(\frac{1}{C(1+\alpha k)}\right)^{1/\alpha}\le a_0,
\]
the benchmark distortion in \eqref{eq:D2_def} satisfies
\begin{equation}\label{eq:D2_small_ball_bound}
D_2(k)\;\ge\; \e^{-1/\alpha}\left(\frac{1}{C(1+\alpha k)}\right)^{1/\alpha}.
\end{equation}
In particular,
\[
D_2(k)=\Omega\!\big(k^{-1/\alpha}\big).
\]
\end{theorem}
\begin{IEEEproof}
See Appendix~\ref{app-new_2}.
\end{IEEEproof}

Theorem~\ref{th_2} is a converse-style result for the decentralized MMSE benchmark: it rules out decay of $D_2(k)$ faster than $k^{-1/\alpha}$ when the single-agent MMSE error satisfies \eqref{eq:ball_assumption}.  

\rev{The proof also applies to any fixed common estimator $g$ whose error satisfies the averaged small-ball condition. Section~\ref{subsec:other_estimators} states this extension and gives a bound uniform over deterministic and randomized local rules under an additional condition on the posterior density.}

A particularly useful coordinate-free sufficient condition is boundedness of the conditional \emph{joint} error density near the origin.

\begin{coro}[Bounded joint error density near the origin]
\label{coro:D2_bd_density}
Assume that conditioned on $\bm X=\bm x$, the error $\bm Z=\bm X-g(\bm Y)$ admits a Lebesgue density $f_{\bm Z|\bm X}(\bm z|\bm x)$ on $\mathbb{R}^d$. Suppose there exist constants $0<M<\infty$ and $r>0$ such that
\begin{equation}\label{eq:bounded_density_assumption}
\sup_{\bm x\in\mathbb{R}^d}\ \sup_{\|\bm z\|_2\le r}\ f_{\bm Z|\bm X}(\bm z|\bm x)\ \le\ M.
\end{equation}
Then, for all $a\in[0,r^2]$ and all $\bm x$,
\[
\mathbb{P}\!\left(\|\bm Z\|_2^2\le a \,\middle|\,\bm X=\bm x\right)\le (M V_d)\,a^{d/2}.
\]
Consequently, Theorem~\ref{th_2} applies with $\alpha=d/2$, $C=M V_d$, and $a_0=r^2$. Therefore, for every $k$ such that
\[
\left(\frac{1}{(M V_d)\left(1+\frac d2\,k\right)}\right)^{2/d}\le r^2,
\]
the benchmark distortion satisfies
\begin{equation}\label{eq:D2_bd_density_bound}
D_2(k)\;\ge\; \e^{-2/d}\left(\frac{1}{(M V_d)\left(1+\frac d2\,k\right)}\right)^{2/d}
\;=\;\Omega\!\big(k^{-2/d}\big).
\end{equation}
\end{coro}
\begin{IEEEproof}
For $a\le r^2$,
\[
\mathbb{P}(\|\bm Z\|_2^2\le a\mid \bm X=\bm x)
=
\int_{\|\bm z\|_2\le \sqrt a} f_{\bm Z|\bm X}(\bm z|\bm x)\,\D\bm z
\le
M\,\mathrm{Vol}(B_d(\sqrt a))
=
M V_d\,a^{d/2}.
\]
Then apply Theorem~\ref{th_2}.
\end{IEEEproof}

Thus, a bounded joint error density near $\bm 0$ forces the decentralized MMSE benchmark to have lower-bound exponent $2/d$.

\begin{remark}[Non-uniform density bounds]
If \eqref{eq:bounded_density_assumption} holds with a bound $M(\bm x)$ depending on $\bm x$, i.e.,
\[
\sup_{\|\bm z\|_2\le r} f_{\bm Z|\bm X}(\bm z|\bm x)\le M(\bm x),
\]
and $0<\mathbb{E}[M(\bm X)]<\infty$, then averaging the resulting conditional small-ball bound over $\bm X$ yields \eqref{eq:ball_assumption} with $C=V_d\,\mathbb{E}[M(\bm X)]$ and $\alpha=d/2$. Hence one still obtains $D_2(k)=\Omega(k^{-2/d})$ without uniformity in $\bm x$.
\end{remark}

\begin{coro}
\label{coro:D2_d1}
For $d=1$, $V_1=2$, and Corollary~\ref{coro:D2_bd_density} gives, for all sufficiently large $k$,
\[
D_2(k)\;\ge\; \e^{-2}\left(\frac{1}{2M\left(1+\frac{k}{2}\right)}\right)^{2}
=\Omega(k^{-2}),
\]
where $0<M<\infty$ satisfies
$\sup_{x}\sup_{|z|\le r} f_{Z|X}(z|x)\le M$.
\end{coro}

The following corollary summarizes the smooth-density regime most relevant for the comparison with the centralized result.

\begin{coro}
\label{coro:compare_smooth}
Under the assumptions of Theorem~\ref{thm_2a} and Corollary~\ref{coro:D2_bd_density}, we have
\[
D_1(k)=\Theta(k^{-2/d}),
\qquad
D_2(k)=\Omega(k^{-2/d}).
\]
Thus, in this smooth regime, the decentralized MMSE benchmark cannot improve upon the centralized $k^{-2/d}$ exponent.
\end{coro}

\subsection{Exponent-Level Comparison in General Dimension}
\label{subsec:d-dim-rate-comparison}

We now compare the decay exponents supplied by the exact asymptotic for $D_1(k)$ and the converse-style lower bounds for $D_2(k)$. Under Theorem~\ref{thm_2a}, the centralized $k$-list estimator has exponent $2/d$. On the decentralized side, the relevant quantity is the small-ball exponent of the single-agent MMSE squared error $W=\|\bm Z\|_2^2$ near zero, for which there are two relevant cases.

\subsubsection*{Case 1: Bounded density near the origin}

Suppose that the conditional density $f_{\bm Z|\bm X}(\bm z|\bm x)$ exists and is uniformly bounded in a neighborhood of $\bm z=\bm 0$, i.e.\ \eqref{eq:bounded_density_assumption} holds for some $0<M<\infty$ and $r>0$. Then Corollary~\ref{coro:D2_bd_density} yields
\[
D_2(k)=\Omega(k^{-2/d}),
\]
while Theorem~\ref{thm_2a} yields
\[
D_1(k)=\Theta(k^{-2/d}).
\]
Hence, in the smooth-density regime, the decentralized MMSE benchmark cannot decay faster than the centralized $k$-list estimator at the exponent level.

\subsubsection*{Case 2: Density vanishing near the origin}

Suppose that there exist $r>0$, $\beta>0$, and a nonnegative function
$c_{\max}(\bm x)$, with
$0<\mathbb{E}[c_{\max}(\bm X)]<\infty$, such that
\begin{equation}\label{eq:hole_density_powerlaw}
f_{\bm Z|\bm X}(\bm z|\bm x)\le c_{\max}(\bm x)\,\|\bm z\|_2^{\beta},
\qquad \|\bm z\|_2\le r.
\end{equation}
Then Lemma~\ref{lem:powerlaw_smallball} in Appendix~\ref{app_2} implies that, for all $a\in[0,r^2]$ and all $\bm x$,
\[
\mathbb{P}\!\left(\|\bm Z\|_2^2\le a \,\middle|\, \bm X=\bm x\right)
\le
\frac{S_d}{d+\beta}\,c_{\max}(\bm x)\,a^{(d+\beta)/2}.
\]
Averaging over $\bm X$ therefore gives the averaged small-ball condition \eqref{eq:ball_assumption} with
\[
\alpha=\frac{d+\beta}{2},
\qquad
C=\frac{S_d}{d+\beta}\,\mathbb{E}[c_{\max}(\bm X)].
\]
Theorem~\ref{th_2} then yields
\[
D_2(k)=\Omega\!\left(k^{-2/(d+\beta)}\right).
\]
Because $\beta>0$, the exponent $2/(d+\beta)$ is strictly smaller than $2/d$. Therefore, in this regime, the benchmark lower bound decays more slowly than the centralized $k^{-2/d}$ law, so the decentralized MMSE benchmark is provably unable to match the centralized exponent.

\rev{Section~\ref{subsec:vanishing_example} gives an informative example in dimension $d$ satisfying this vanishing-density condition and applies Theorems~\ref{thm_2a} and~\ref{th_2} to obtain a strict separation between the exponents of $D_1(k)$ and $D_2(k)$ as $k$ grows.}

If, in addition, matching lower bounds of the form
\[
c_{\min}(\bm x)\,\|\bm z\|_2^\beta
\le
f_{\bm Z|\bm X}(\bm z|\bm x)
\le
c_{\max}(\bm x)\,\|\bm z\|_2^\beta
\]
hold near the origin, then, whenever
$0<c_{\min}(\bm x)\le c_{\max}(\bm x)<\infty$,
Lemma~\ref{lem:powerlaw_smallball} shows that the conditional
small-ball probability itself has exponent $(d+\beta)/2$.

\begin{remark}[Outside-scope regimes]
The density-based sufficient conditions above can fail when the MMSE error distribution has atoms, is singular with respect to Lebesgue measure, or exhibits integrable singularities at the origin. Such cases require a separate analysis and lie outside the continuous smooth-density regime emphasized in this paper.
\end{remark}

\begin{remark}[Intuition for Cases~1--2]
In Case~1, a bounded conditional error density near $\bm 0$ means that the probability of an exceptionally small error is essentially limited by the volume of a radius-$\sqrt a$ ball in $\mathbb{R}^d$, leading to the exponent $d/2$. In Case~2, the density itself vanishes near $\bm 0$, which suppresses very small errors even further and increases the small-ball exponent to $(d+\beta)/2$. Through Theorem~\ref{th_2}, a larger small-ball exponent translates into a slower converse bound for the decentralized benchmark.
\end{remark}

The comparison can be summarized as
\begin{equation}\label{eq:exponent_summary}
\begin{aligned}
D_1(k) &= \Theta\!\big(k^{-2/d}\big),\\
D_2(k) &= \Omega\!\big(k^{-2/d}\big) \qquad &&\text{under \eqref{eq:bounded_density_assumption}},\\
D_2(k) &= \Omega\!\big(k^{-2/(d+\beta)}\big) \qquad &&\text{under \eqref{eq:hole_density_powerlaw} with }\beta>0.
\end{aligned}
\end{equation}
Hence, the centralized $k$-list estimator always has exponent $2/d$ under Theorem~\ref{thm_2a}. In the smooth-density regime, the decentralized MMSE benchmark cannot improve upon that exponent, while in vanishing-density regimes its converse exponent is strictly worse.


\begin{revision}
\subsection{Other Decentralized Estimators}\label{subsec:other_estimators}

For a common measurable rule $g:\mathbb R^m\to\mathbb R^d$ that need not be MMSE, define
\begin{equation}\label{eq:D2_general_rule}
D_{2,g}(k)\triangleq\mathbb{E}\!\left[
\min_{1\le i\le k}\|\bm X-g(\bm Y_i)\|_2^2\right].
\end{equation}
The small-ball argument of Theorem~\ref{th_2} gives the following extension.

\begin{coro}\label{coro:D2_general_rule}
Assume that \eqref{eq:agent_cond_iid_factorization} holds and that the error of $g$ satisfies
\[
\mathbb P\!\left(\|\bm X-g(\bm Y)\|_2^2\le a\right)
\le Ca^\alpha,\qquad a\in[0,a_0],
\]
for constants $C>0$, $\alpha>0$, and $a_0>0$. Then, for every $k$ such that
$a^\star=[C(1+\alpha k)]^{-1/\alpha}\le a_0$,
\[
D_{2,g}(k)\ge \e^{-1/\alpha}
\left(\frac{1}{C(1+\alpha k)}\right)^{1/\alpha}
=\Omega\!\big(k^{-1/\alpha}\big).
\]
\end{coro}
\begin{IEEEproof}
Under \eqref{eq:agent_cond_iid_factorization}, the errors induced by $g$ are conditionally i.i.d.\ given $\bm X$. The proof of Theorem~\ref{th_2} in Appendix~\ref{app-new_2} therefore applies without using the MMSE property.
\end{IEEEproof}

The density conditions in Corollary~\ref{coro:D2_bd_density} and Section~\ref{subsec:d-dim-rate-comparison} likewise apply to the error induced by the estimation rule $g$. These conditions must be checked for the adopted rule $g$ since strict separation due to a vanishing MMSE error density need not persist for the adopted rule.

A condition on the observation model gives a lower bound uniformly over all local estimation rules, $g_i$, including randomized ones, as stated in the following proposition.

\begin{proposition}\label{prop:general_local_bound}
Suppose that the posterior density satisfies
\begin{equation}\label{eq:posterior_sup_mean}
\overline M\triangleq\mathbb{E}_{\bm Y}\!\left[
\operatorname*{ess\,sup}_{\bm x\in\mathbb{R}^d}
f_{\bm X|\bm Y}(\bm x|\bm Y)\right]<\infty.
\end{equation}
Let agent $i$ output $g_i(\bm Y_i,U_i)$, where the measurable mappings may differ across agents and the seed vector $(U_1,\ldots,U_k)$ is independent of $(\bm X,\bm Y_1,\ldots,\bm Y_k)$. The seeds may be private, shared, or correlated with one another. Then, for every $k\ge1$,
\begin{equation}\label{eq:general_local_bound}
\mathbb{E}\!\left[\min_{1\le i\le k}
\|\bm X-g_i(\bm Y_i,U_i)\|_2^2\right]
\ge \frac{d}{d+2}(kV_d\overline M)^{-2/d}.
\end{equation}
The mappings and seed distributions may depend on $k$.

Deterministic rules are included  by taking constant seeds.
\end{proposition}
\begin{IEEEproof}
See Appendix~\ref{app:general_local_proof}.
\end{IEEEproof}

 Together with Theorem~\ref{thm_2a}, Proposition~\ref{prop:general_local_bound} shows that, under their respective assumptions, the centralized list estimation achieves a distortion decay exponent no worse than that of any such decentralized local design. Since the bound is uniform over the local estimation rules, $g_i$, it also applies to the infimum over these designs. Different rules may have different decay exponents or leading constants, but none can attain a decay exponent better than $2/d$.
\end{revision}

\section{\rev{Specializations and Examples}}\label{sec:gaussian_examples}

Gaussian models provide a canonical smooth-density setting in which the general results of Section~\ref{sec:main_results} become fully explicit. On the centralized side, the Gaussian posterior integral appearing in the high-rate formula can be evaluated in closed form. On the decentralized side, the conditional MMSE error is Gaussian and therefore has a smooth, bounded density near the origin, so the coordinate-free lower bound of Corollary~\ref{coro:D2_bd_density} applies automatically. The resulting formulas also serve as the theoretical reference predictions for the numerical experiments in Section~\ref{sec:simres}.
\rev{The final subsection gives a non-Gaussian example of the vanishing-density regime.}

\subsection{Closed-form evaluation of the Gaussian posterior functional}

For a Gaussian density, the Zador functional depends only on the covariance matrix; the mean plays no role.

\begin{lem}[Zador functional for a Gaussian density]
\label{lem:zador_gaussian}
Let $f$ be the density of $\mathcal{N}(\bm \mu,\bm \Sigma)$ on $\mathbb{R}^d$ with $\bm \Sigma\succ 0$. Let $p=\frac{d}{d+2}$. Then
\begin{equation}\label{eq:zador_gaussian_integral}
\left(\int_{\mathbb{R}^d} f(\bm x)^p\,\D\bm x\right)^{\!\frac{d+2}{d}}
=
(2\pi)\,\det(\bm \Sigma)^{1/d}\left(\frac{d+2}{d}\right)^{\!\frac{d+2}{2}}.
\end{equation}
\end{lem}
\begin{IEEEproof}
Write
\[
f(\bm x)=(2\pi)^{-d/2}\det(\bm \Sigma)^{-1/2}
\exp\!\left(-\frac12(\bm x-\bm \mu)^\top\bm \Sigma^{-1}(\bm x-\bm \mu)\right).
\]
Then
\[
f(\bm x)^p=(2\pi)^{-dp/2}\det(\bm \Sigma)^{-p/2}
\exp\!\left(-\frac{p}{2}(\bm x-\bm \mu)^\top\bm \Sigma^{-1}(\bm x-\bm \mu)\right).
\]
The exponential term is proportional to a Gaussian density with precision matrix $p\bm \Sigma^{-1}$, so
\[
\int_{\mathbb{R}^d}\exp\!\left(-\frac{p}{2}(\bm x-\bm \mu)^\top\bm \Sigma^{-1}(\bm x-\bm \mu)\right)\D\bm x
=(2\pi)^{d/2}\det(\bm \Sigma)^{1/2}p^{-d/2}.
\]
Therefore,
\[
\int_{\mathbb{R}^d} f(\bm x)^p\,\D\bm x
=
(2\pi)^{-dp/2}\det(\bm \Sigma)^{-p/2}
\cdot
(2\pi)^{d/2}\det(\bm \Sigma)^{1/2}p^{-d/2}
=
(2\pi)^{\frac{d(1-p)}{2}}
\det(\bm \Sigma)^{\frac{1-p}{2}}
p^{-d/2}.
\]
Since $1-p=\frac{2}{d+2}$, raising the above quantity to the power $\frac{d+2}{d}$ yields \eqref{eq:zador_gaussian_integral}.
\end{IEEEproof}

Thus, for Gaussian posteriors, the leading centralized constant depends on the posterior covariance only through $\det(\bm \Sigma)^{1/d}$, i.e., the geometric mean of its eigenvalues.

\begin{coro}[Explicit Gaussian specialization of the centralized high-rate coefficient]
\label{coro:D1_gaussian_constant}
Assume that for almost every $\bm y$, the posterior $\bm X|\bm Y=\bm y$ is Gaussian with covariance matrix $\bm \Sigma_{\bm X|\bm Y}(\bm y)\succ 0$. \rev{Under the assumptions of Theorem~\ref{thm_2a},} Theorem~\ref{thm_2a} and Lemma~\ref{lem:zador_gaussian} yield
\begin{equation}\label{eq:D1_gaussian_constant}
D_1(k)
=
G_d\,k^{-2/d}\,(2\pi)\,
\mathbb{E}_{\bm Y}\!\left[
\det\!\left(\bm \Sigma_{\bm X|\bm Y}(\bm Y)\right)^{1/d}
\left(\frac{d+2}{d}\right)^{\!\frac{d+2}{2}}
\right]
+
o(k^{-2/d}).
\end{equation}
In particular, if $\bm \Sigma_{\bm X|\bm Y}(\bm y)$ is constant in $\bm y$ (as in standard linear Gaussian models), then the expectation over $\bm Y$ drops.
\end{coro}
\begin{IEEEproof}
Apply Lemma~\ref{lem:zador_gaussian} pointwise in $\bm y$ to the posterior density $f_{\bm X|\bm Y}(\cdot|\bm y)$ and substitute the result into \eqref{eq:general_constant}.
\end{IEEEproof}

In general dimension, \eqref{eq:D1_gaussian_constant} is explicit but not fully closed-form, because it still depends on the quantization constant $G_d$, and the expectation over $\bm Y$, which may not simplify further unless the posterior covariance is observation-independent. In the standard linear Gaussian model, the expectation drops; in the scalar case $d=1$, one additionally has $G_1=1/12$, yielding a fully closed-form constant.

\subsection{Isotropic additive Gaussian observation model}\label{sec-pred}

We now specialize to the standard additive Gaussian model used later in the simulations. The isotropic form is chosen for concreteness and for closed-form evaluation; the general decentralized lower bound itself does not rely on coordinate-wise independence.

Consider
\[
\bm X\sim\mathcal{N}(\bm 0,\sigma_X^2\mathbf{I}),
\qquad
\bm N_i\sim\mathcal{N}(\bm 0,\sigma_N^2\mathbf{I}),
\]
where $\bm N_1,\ldots,\bm N_k$ are independent across $i$ and independent of $\bm X$, and
\[
\bm Y_i=\bm X+\bm N_i,\qquad i\in[1:k].
\]
The single-agent MMSE estimator is linear:
\begin{equation}\label{eq_2b}
g(\bm y)=\mathbb{E}[\bm X|\bm Y=\bm y]
=\frac{\sigma_X^2}{\sigma_X^2+\sigma_N^2}\bm y.
\end{equation}

\noindent\textbf{Decentralized MMSE benchmark:}
Using \eqref{eq_2b},
\[
\bm Z_i\triangleq \bm X-g(\bm Y_i)
=\bm X-\frac{\sigma_X^2}{\sigma_X^2+\sigma_N^2}(\bm X+\bm N_i)
=
\frac{\sigma_N^2}{\sigma_X^2+\sigma_N^2}\bm X
-
\frac{\sigma_X^2}{\sigma_X^2+\sigma_N^2}\bm N_i.
\]
Conditioned on $\bm X=\bm x$, the error $\bm Z_i$ is Gaussian with mean
\[
\bm \mu_{\bm Z|\bm X=\bm x}
=
\frac{\sigma_N^2}{\sigma_X^2+\sigma_N^2}\bm x
\]
and covariance $\sigma_G^2\mathbf{I}$, where
\begin{equation}\label{eq:sigmaG_scalar}
\sigma_G^2=
\frac{\sigma_X^4\sigma_N^2}{(\sigma_X^2+\sigma_N^2)^2}.
\end{equation}
Hence the conditional error density is smooth and uniformly bounded, with
\[
\sup_{\bm z,\bm x} f_{\bm Z|\bm X}(\bm z|\bm x)
=
(2\pi\sigma_G^2)^{-d/2}
\triangleq M.
\]
Applying Corollary~\ref{coro:D2_bd_density} gives
\begin{align}\label{eq:D2_gaussian_general_lower}
D_2(k)
&\ge \e^{-2/d}\left(\frac{1}{(M V_d)\left(1+\frac d2\,k\right)}\right)^{2/d}\nonumber\\
&= \e^{-2/d}\left(\frac{(2\pi\sigma_G^2)^{d/2}}{V_d\left(1+\frac d2\,k\right)}\right)^{2/d}\nonumber\\
&= \e^{-2/d}\,\frac{2\pi\sigma_G^2}{V_d^{2/d}\left(1+\frac d2\,k\right)^{2/d}}
=\Omega(k^{-2/d}).
\end{align}

\noindent\textbf{Centralized $k$-list estimator:}
For each observation $\bm y$, the posterior is Gaussian:
\[
\bm X|\bm Y=\bm y \sim \mathcal{N}\!\left(\bm \mu_{\bm X|\bm Y=\bm y},\bm \Sigma_{\bm X|\bm Y}\right),
\]
with
\[
\bm \mu_{\bm X|\bm Y=\bm y}
=
\frac{\sigma_X^2}{\sigma_X^2+\sigma_N^2}\bm y,
\qquad
\bm \Sigma_{\bm X|\bm Y}
=
\sigma_{X|Y}^2\mathbf{I},
\qquad
\sigma_{X|Y}^2
=
\frac{\sigma_X^2\sigma_N^2}{\sigma_X^2+\sigma_N^2}.
\]
Since the posterior covariance does not depend on $\bm y$, Corollary~\ref{coro:D1_gaussian_constant} yields
\begin{equation}\label{eq:D1_gaussian_general_constant}
D_1(k)=
G_d\,k^{-2/d}\,(2\pi)\,\sigma_{X|Y}^2
\left(\frac{d+2}{d}\right)^{\!\frac{d+2}{2}}
+o(k^{-2/d}).
\end{equation}

\begin{remark}[Translation structure of the Gaussian posterior]
Because $\bm \Sigma_{\bm X|\bm Y}$ is constant in $\bm y$ and $\bm \mu_{\bm X|\bm Y=\bm y}$ is affine in $\bm y$, the posterior at any observation $\bm y$ is a translation of the zero-mean posterior at $\bm y=\bm 0$. Consequently, an optimal or approximately optimal posterior list for $\bm y$ can be obtained by translating the corresponding list for $\bm y=\bm 0$ by $\bm \mu_{\bm X|\bm Y=\bm y}$. This property is exploited in Section~\ref{sec:simres} to reduce the computational cost of the $k$-means approximation.
\end{remark}

Therefore, the additive Gaussian model lies squarely in the smooth-density regime of Section~\ref{subsec:d-dim-rate-comparison}: the centralized distortion obeys the exact $k^{-2/d}$ law in \eqref{eq:D1_gaussian_general_constant}, while the decentralized MMSE benchmark is lower-bounded by a term of order $k^{-2/d}$ and hence cannot improve upon the centralized exponent. \rev{Moreover, $\overline M=(2\pi\sigma_{X|Y}^2)^{-d/2}<\infty$, so Proposition~\ref{prop:general_local_bound} extends this exponent comparison to every local estimation rule covered by that proposition.}

\subsection{Scalar Gaussian specialization ($d=1$)}\label{sec-G1}

For $d=1$, the general formulas above simplify to fully closed-form scalar expressions, because $G_1=1/12$. Since $G_1=1/12$ and $V_1=2$, \eqref{eq:D1_gaussian_general_constant} becomes
\[
D_1(k)=\frac{\sqrt{3}\pi}{2}\,\sigma_{X|Y}^2\,\frac{1}{k^2}+o(k^{-2}),
\]
where
\[
\sigma_{X|Y}^2=\frac{\sigma_X^2\sigma_N^2}{\sigma_X^2+\sigma_N^2}.
\]
On the decentralized side, \eqref{eq:D2_gaussian_general_lower} yields
\begin{equation}\label{eq_dv2}
D_2(k)\ge \e^{-2}\left(\frac{1}{(2M)\left(1+\frac{k}{2}\right)}\right)^2
=\frac{2\pi\sigma_G^2}{\e^2(k+2)^2},
\end{equation}
where
\[
M=(2\pi\sigma_G^2)^{-1/2},
\qquad
\sigma_G^2=\frac{\sigma_X^4\sigma_N^2}{(\sigma_X^2+\sigma_N^2)^2}.
\]
Thus, in the scalar Gaussian model, the centralized best-candidate distortion decays as $k^{-2}$, while the decentralized MMSE benchmark is lower-bounded by a term of order $k^{-2}$. In particular, the benchmark cannot have a better $k$-exponent than the centralized estimator.


\begin{revision}
\subsection{An informative vanishing-density example}\label{subsec:vanishing_example}

The strict separation of the decay exponents of $D_1(k)$ and $D_2(k)$, established in Section~\ref{subsec:d-dim-rate-comparison}, can occur with informative observations in every dimension.

 Let $\varphi_d(\bm z)=(2\pi)^{-d/2}\exp(-\|\bm z\|_2^2/2)$, let $\bm Y\sim\mathcal N(\bm 0,\mathbf I)$, and let $\bm Z$ be independent of $\bm Y$ with density
\begin{equation}\label{eq:vanishing_example_density}
h_d(\bm z)=\frac{\|\bm z\|_2^2}{d}\varphi_d(\bm z),
\end{equation}
which is the spherically symmetric Kotz-type distribution \cite{Kotz-distr}.
The scaling $1/d$ normalizes the density, and symmetry gives $\mathbb E[\bm Z]=\bm 0$.
The factor $\|\bm z\|_2^2$ makes the density vanish quadratically
at the origin, unlike a Gaussian density, and is the source of
the strict exponent separation established below.
 
 Set $\bm X=\bm Y+\bm Z$. Then
\[
f_{\bm X|\bm Y}(\bm x|\bm y)=h_d(\bm x-\bm y),\qquad
g(\bm y)=\mathbb E[\bm X\mid\bm Y=\bm y]=\bm y.
\]
The observation is informative, since $\operatorname{Cov}(\bm X,\bm Y)=\mathbf I$. This joint law defines the observation kernel
\begin{equation}\label{eq:vanishing_example_kernel}
\begin{aligned}
f_{\bm Y|\bm X}(\bm y|\bm x)
&=\frac{\varphi_d(\bm y)h_d(\bm x-\bm y)}{f_{\bm X}(\bm x)},\\
f_{\bm X}(\bm x)
&=\int_{\mathbb R^d}\varphi_d(\bm y)h_d(\bm x-\bm y)\,\D\bm y>0.
\end{aligned}
\end{equation}
The decentralized observations are conditionally i.i.d.\ given $\bm X$, using this same kernel.

The posterior densities are translations of the smooth, bounded density $h_d$, whose Zador functional is finite and positive. Moreover, $\bm X$ has finite moments of every order, which also implies the averaged moment condition in Theorem~\ref{thm_2a}. Thus the assumptions of that theorem hold, and
\begin{equation}\label{eq:vanishing_example_D1}
\begin{aligned}
D_1(k)
&=G_d\left(\int_{\mathbb R^d}h_d(\bm z)^{d/(d+2)}\,\D\bm z\right)^{\!\frac{d+2}{d}}
k^{-2/d}+o(k^{-2/d})\\
&=\Theta(k^{-2/d}).
\end{aligned}
\end{equation}
For the MMSE rule, the single-agent error is $\bm X-g(\bm Y)=\bm Z$, and
\begin{equation}\label{eq:vanishing_example_smallball}
\begin{aligned}
\mathbb P[\|\bm Z\|_2^2\le a]
&=\frac1d\int_{\|\bm z\|_2\le\sqrt a}\|\bm z\|_2^2\varphi_d(\bm z)\,\D\bm z\\
&\le C_d a^{(d+2)/2},\qquad a\ge0,
\end{aligned}
\end{equation}
where $C_d=V_d/[(d+2)(2\pi)^{d/2}]$. Applying Theorem~\ref{th_2} with $\alpha=(d+2)/2$ and $C=C_d$ gives
\begin{equation}\label{eq:vanishing_example_D2}
D_2(k)\ge \e^{-2/(d+2)}
\left(\frac{1}{C_d(1+(d+2)k/2)}\right)^{\!2/(d+2)}
=\Omega(k^{-2/(d+2)}).
\end{equation}
Consequently, $D_1(k)/D_2(k)\to0$ as $k\to\infty$. The conditional error density also satisfies the vanishing-density condition in Section~\ref{subsec:d-dim-rate-comparison} with $\beta=2$, as verified at the end of Appendix~\ref{app_2}. For $d=1$, the construction reduces to $h_1(z)=z^2\varphi_1(z)$.

This example also satisfies
\[
\overline M=\sup_{\bm z\in\mathbb R^d}h_d(\bm z)
=\frac{2}{d\e(2\pi)^{d/2}}<\infty.
\]
Hence Proposition~\ref{prop:general_local_bound} limits the decay exponent of the decentralized benchmark for every admissible local
estimation rule to $2/d$, whereas \eqref{eq:vanishing_example_D2}
gives the stricter limit $2/(d+2)$ for the MMSE decentralized benchmark.
\end{revision}

\section{Numerical Validation}\label{sec:simres}

This section numerically validates the exponent-level predictions developed in Sections~\ref{sec:main_results} and \ref{sec:gaussian_examples}. We estimate the centralized best-candidate distortion $D_1(k)$ and the decentralized MMSE benchmark distortion $D_2(k)$ in the isotropic additive Gaussian model of Section~\ref{sec-pred}, compare $D_2(k)$ with the theoretical lower bound in \eqref{eq:D2_gaussian_general_lower}, and compare the empirical centralized curve $D_1(k)$  with the high-rate prediction in \eqref{eq:D1_gaussian_general_constant}. The emphasis is on scaling with $k$, rather than on exact constant matching over a finite range of list sizes.

\subsection{Simulation setup and reproducibility details}

We consider the additive Gaussian model
\[
\bm X\sim\mathcal{N}(\bm 0,\mathbf{I}),\qquad
\bm Y_i=\bm X+\bm N_i,
\]
where $\bm N_i\sim\mathcal{N}(\bm 0,\sigma_N^2\mathbf{I})$ independently across $i$ and independently of $\bm X$. We study noise levels $\sigma_N\in\{0.2,1,5\}$ and dimensions $d\in\{1,4,10\}$ over the plotted range $1\le k\le 10^3$.

For each configuration $(d,\sigma_N,k)$, the decentralized MMSE benchmark distortion $D_2(k)$ is estimated by Monte Carlo from \eqref{eq:D2_def}, using the linear MMSE estimator in \eqref{eq_2b}. The results reported here use $10^5$ Monte Carlo trials for the benchmark curves. The corresponding theoretical lower bound is computed from \eqref{eq:D2_gaussian_general_lower}.

The centralized distortion $D_1(k)$ is approximated numerically by sampled $k$-means.\footnote{The empirical version of the optimal $k$-point
quantization problem is the Euclidean $k$-means problem,
which is NP-hard in general \cite[Theorem~1]{Aloise2009}.
We use the $k$-means algorithm as a numerical approximation.} The experiments are implemented in Python using the FAISS library with GPU-accelerated $k$-means \cite{Johnson_2019_1}.  

\begin{remark}[Approximating $G_d$ in the high-rate reference curve]
The centralized high-rate prediction in \eqref{eq:D1_gaussian_general_constant} depends on the constant $G_d$. For $d=1$, this value is exact: $G_1=1/12$. For $d=4$ and $d=10$, where no closed-form expression is available, we use the best reported lattice normalized second moments (NSMs) from \cite[Table~I]{AgrellAllen2023}. Specifically, if
\[
G(\Lambda)\triangleq \frac{1}{d\,\mathrm{Vol}(V)^{1+2/d}}\int_V \|\bm z\|_2^2\,d\bm z
\]
denotes the NSM of a lattice $\Lambda$ with Voronoi cell $V$, then the distortion constant in our notation is approximated by
\[
G_d \approx d\,G(\Lambda_d^{\rm best}).
\]

Accordingly, in the plotted high-rate reference curves we use
\[
G_1=\frac{1}{12},\quad
G_4\approx 4\times 0.076603235 = 0.30641294,\quad
G_{10}\approx 10\times 0.070813818 = 0.70813818,
\]
where $0.076603235$ and $0.070813818$ are the Table~I NSM values for $D_4$ and $D_{10}^{+}$, respectively \cite{AgrellAllen2023}. As an additional asymptotic check, one may also compare
with the large-$d$ proxy $G_d\approx d(2\pi e)^{-1}$, based on
the convergence of the optimal lattice normalized second
moment to $(2\pi e)^{-1}$ as $d\to\infty$
\cite[Lemma~1]{ZamirFeder1996} and the above normalization
convention.
\end{remark}

\subsection{Centralized $k$-means approximation and Gaussian shift structure}

In the Gaussian model, the posterior $\bm X|\bm Y=\bm y$ is
\[
\mathcal{N}\!\left(\bm \mu_{\bm X|\bm Y=\bm y},\,\sigma_{X|Y}^2\mathbf{I}\right),
\qquad
\bm \mu_{\bm X|\bm Y=\bm y}
=
\frac{\sigma_X^2}{\sigma_X^2+\sigma_N^2}\bm y,
\qquad
\sigma_{X|Y}^2
=
\frac{\sigma_X^2\sigma_N^2}{\sigma_X^2+\sigma_N^2}.
\]
Because the posterior covariance is independent of $\bm y$, the entire posterior family is obtained by translating a fixed zero-mean Gaussian. We therefore approximate the posterior-optimal list only once, for the zero-mean posterior, by running $k$-means on samples from $\mathcal{N}(\bm 0,\sigma_{X|Y}^2\mathbf{I})$. For a realized observation $\bm y$, the corresponding list is obtained by translation:
\[
\bm x_i(\bm y)=\bm \mu_{\bm X|\bm Y=\bm y}+\bm x_i(\bm 0),\qquad i\in[1:k].
\]
Thus, the Gaussian shift structure is used exactly, while the only approximation comes from the sampled $k$-means solution itself.

The resulting empirical estimate of $D_1(k)$ is obtained by evaluating the min-of-$k$ squared error of the translated list under the posterior distribution. The high-rate reference curve for $D_1(k)$ is computed from \eqref{eq:D1_gaussian_general_constant}, using $G_1=1/12$ when $d=1$ and the corrected lattice-based proxies $G_4\approx 0.30641294$ and $G_{10}\approx 0.70813818$ when $d\in\{4,10\}$.

\subsection{Results and discussion}

\begin{figure}[ht!]
 \centering
  \includegraphics[width=1\linewidth]{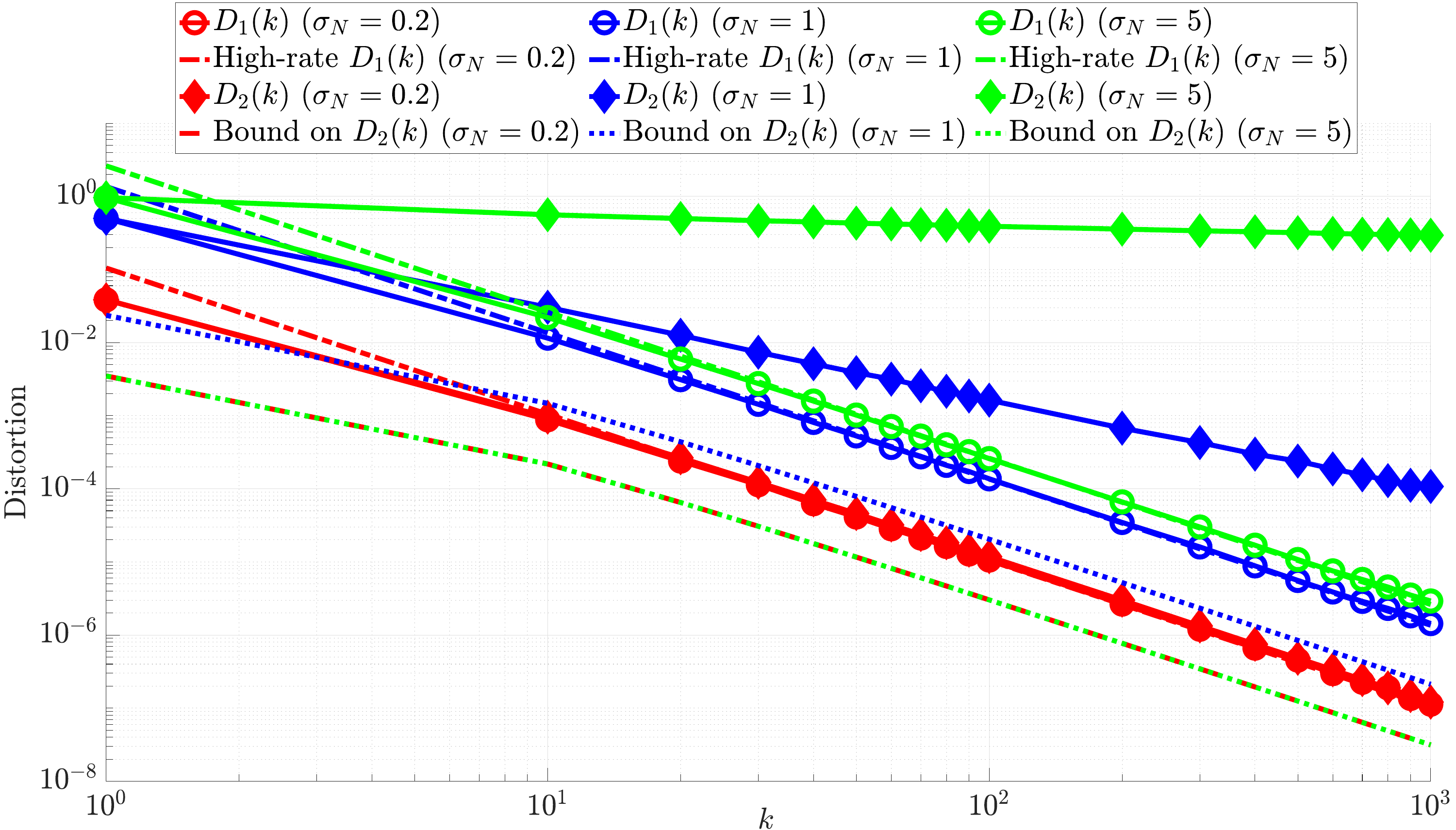}
  \caption{Empirical estimates of the centralized distortion $D_1(k)$ and the decentralized MMSE benchmark distortion $D_2(k)$, together with the theoretical lower bound on $D_2(k)$ and the high-rate prediction for $D_1(k)$, for $d=1$, $\sigma_X=1$, and $\sigma_N\in\{0.2,1,5\}$.}
  \label{fig_1}
\end{figure}

\begin{figure}[ht!]
 \centering
  \includegraphics[width=1\linewidth]{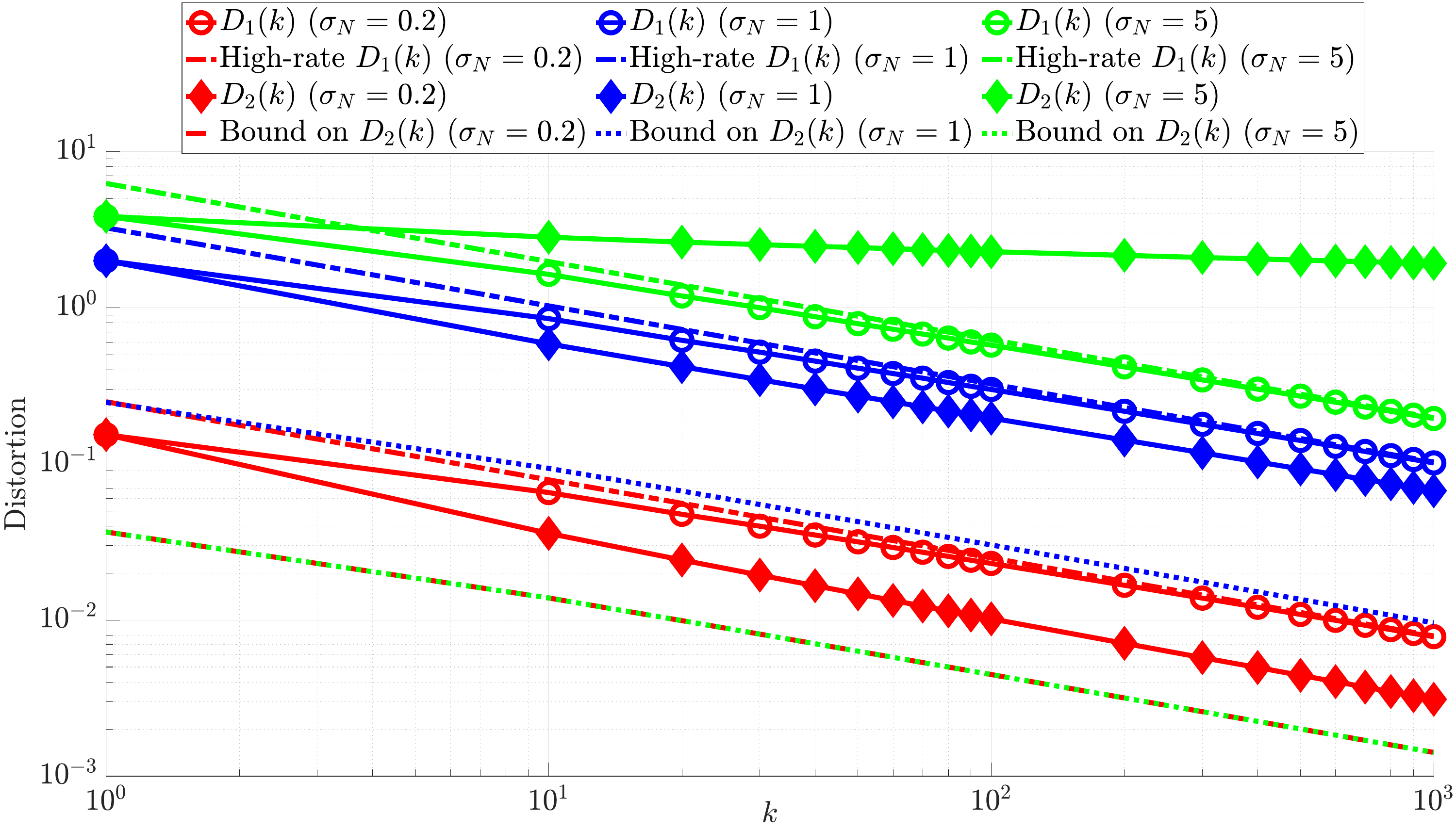}
  \caption{Empirical estimates of the centralized distortion $D_1(k)$ and the decentralized MMSE benchmark distortion $D_2(k)$, together with the theoretical lower bound on $D_2(k)$ and the high-rate prediction for $D_1(k)$, for $d=4$, $\sigma_X=1$, and $\sigma_N\in\{0.2,1,5\}$.}
  \label{fig_2}
\end{figure}

\begin{figure}[ht!]
 \centering
  \includegraphics[width=1\linewidth]{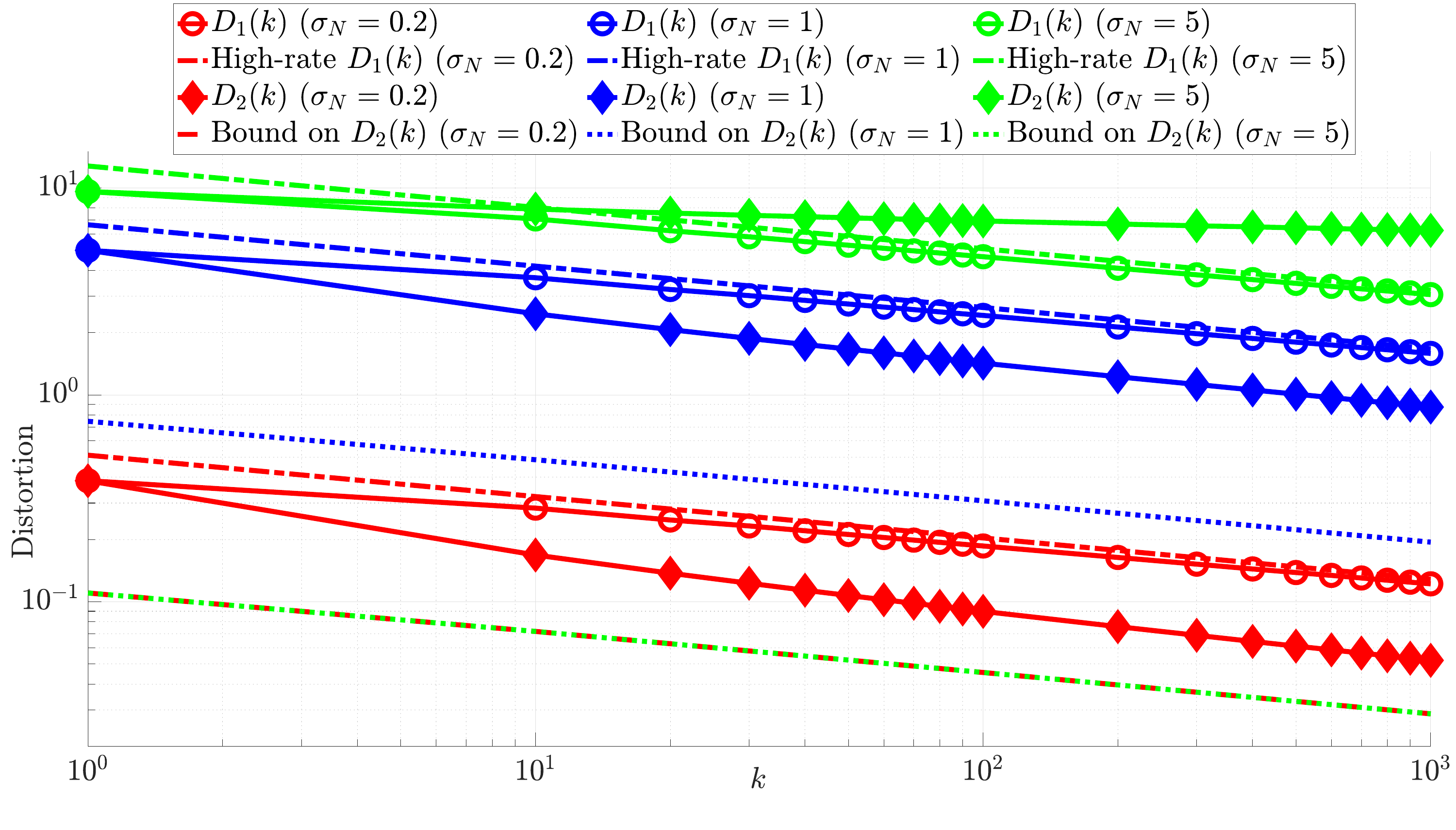}
  \caption{Empirical estimates of the centralized distortion $D_1(k)$ and the decentralized MMSE benchmark distortion $D_2(k)$, together with the theoretical lower bound on $D_2(k)$ and the high-rate prediction for $D_1(k)$, for $d=10$, $\sigma_X=1$, and $\sigma_N\in\{0.2,1,5\}$.}
  \label{fig_3}
\end{figure}

Figures~\ref{fig_1}--\ref{fig_3} confirm the qualitative predictions of the theory. The empirical centralized curves exhibit approximately linear behavior on the log-log scale in the larger-$k$ regime, with slopes consistent with the predicted exponent $-2/d$. The decay is steepest for $d=1$ and becomes progressively shallower for $d=4$ and $d=10$, exactly as dictated by the high-rate law $k^{-2/d}$. The overlaid high-rate curves provide a useful asymptotic reference and track the empirical centralized distortion $D_1(k)$ increasingly well as $k$ grows.

The decentralized MMSE benchmark $D_2(k)$ remains above the theoretical lower bound in all tested settings, as required. \rev{The empirical curves can decay more slowly than the lower bound on $D_2(k)$, and the gap can be substantial, especially at lower SNR and in higher dimension. This is consistent with \eqref{eq:D2_gaussian_general_lower}, which rules out decay faster than order $k^{-2/d}$ but does not assert that $D_2(k)$ attains that order. The lower bound on $D_2(k)$ need not be tight to establish the exponent
comparison with the centralized $k^{-2/d}$ high-rate law of list estimation, $D_1(k)$.}

\rev{The figures support the theoretically established decay rate of the centralized list estimation distortion, $D_1(k)$, and are consistent with the converse on the distortion of the decentralized MMSE benchmark, $D_2(k)$.} The decentralized MMSE benchmark uses $k$ \rev{conditionally independent} observations, whereas the centralized $k$-list estimator uses only one. Consequently, for moderate $k$ and favorable noise levels, the absolute value of the empirical benchmark distortion can be smaller than that of the centralized list estimator. \rev{The theoretical point is that, despite this informational disadvantage, the centralized list estimator achieves an exponent at least as good as that permitted for the decentralized benchmark in the smooth Gaussian regime.}

\begin{remark}[Why the bounds for $\sigma_N=0.2$ and $\sigma_N=5$ coincide]
With $\sigma_X=1$, the quantity governing the Gaussian lower bound,
\[
\sigma_G^2=\frac{\sigma_X^4\sigma_N^2}{(\sigma_X^2+\sigma_N^2)^2}
=\frac{\sigma_N^2}{(1+\sigma_N^2)^2},
\]
takes the same value at $\sigma_N=0.2$ and $\sigma_N=5$. Hence the theoretical lower bound in \eqref{eq:D2_gaussian_general_lower} is identical for these two noise levels. \rev{The corresponding empirical benchmark distortions differ substantially, illustrating that a common lower bound on $D_2(k)$ does not imply equal leading constants or equal actual decay rates for the  decentralized MMSE benchmark, $D_2(k)$.}
\end{remark}

Overall, the simulations support the main theoretical message of the paper: the centralized $k$-list estimator exhibits the predicted $k^{-2/d}$ scaling, and the decentralized MMSE benchmark behaves consistently with the converse bounds derived from the small-ball analysis.


\FloatBarrier
\section{Conclusion}\label{sec_conc}

This paper studied $k$-list estimation, in which a single observation is used to generate $k$ candidate estimates and performance is measured by the squared error of the best candidate in the list. The main conclusion is that, under standard regularity conditions, one observation plus a well-designed list of $k$ candidates \rev{achieves a $k$-exponent at least as good as that of} a natural decentralized MMSE benchmark that uses $k$ conditionally i.i.d.\ observations. More precisely, the centralized best-candidate distortion admits an exact high-rate asymptotic of order $k^{-2/d}$, while the decentralized MMSE benchmark satisfies converse-style lower bounds determined by the small-ball behavior of the single-agent MMSE error. In the smooth-density regime, the benchmark cannot decay faster than order $k^{-2/d}$; when the conditional error density vanishes near the origin in a power-law fashion, the benchmark converse exponent is strictly worse. For Gaussian models, the centralized high-rate coefficient becomes explicit (and fully closed-form in the scalar case), and the numerical results support the predicted exponent-level behavior.

Several directions remain open. The most natural is to characterize the \rev{optimal decentralized mappings and exact distortion} for the min-of-$k$ objective, beyond \rev{the exponent bounds established here}. It would also be of interest to extend the analysis to mixed discrete--continuous, singular, or otherwise non-smooth error laws, and to develop sharper finite-$k$ bounds and constant-level comparisons between centralized and decentralized designs.

\appendices

\section{Proof of Theorem~\ref{th_2}}\label{app-new_2}

This appendix proves Theorem~\ref{th_2}. The key observation is that conditional independence across agents and Jensen's inequality reduce the argument to the averaged small-ball probability of the single-agent MMSE error.

Recall the MMSE estimator
\[
g(\bm y)=\mathbb{E}[\bm X\mid \bm Y=\bm y],
\]
and the quantities
\[
\bm Z_i=\bm X-g(\bm Y_i),\qquad
W_i=\|\bm Z_i\|_2^2,\qquad
\Lambda_k=\min_{1\le i\le k}W_i.
\]
Then $D_2(k)=\mathbb{E}[\Lambda_k]$.

Fix $\bm x\in\mathbb{R}^d$. By the conditional-i.i.d.\ factorization \eqref{eq:agent_cond_iid_factorization}, the observations $\bm Y_1,\ldots,\bm Y_k$ are i.i.d.\ conditioned on $\bm X=\bm x$. Since $g(\cdot)$ is deterministic, the random variables $W_1,\ldots,W_k$ are also i.i.d.\ conditioned on $\bm X=\bm x$.

Now fix $a>0$. Since $\Lambda_k\ge 0$,
\[
\Lambda_k \ge a\,\mathds{1}\{\Lambda_k>a\},
\]
and therefore
\begin{equation}\label{eq:trunc_bound}
\mathbb{E}[\Lambda_k\mid \bm X=\bm x]
\ge
a\,\mathbb{P}(\Lambda_k>a\mid \bm X=\bm x).
\end{equation}
Let
\[
F_{\bm x}(a)\triangleq \mathbb{P}(W_1\le a\mid \bm X=\bm x).
\]
Using conditional independence,
\[
\mathbb{P}(\Lambda_k>a\mid \bm X=\bm x)
=
\mathbb{P}\!\left(\bigcap_{i=1}^k \{W_i>a\}\,\middle|\,\bm X=\bm x\right)
=
\prod_{i=1}^k \mathbb{P}(W_i>a\mid \bm X=\bm x)
=
\big(1-F_{\bm x}(a)\big)^k.
\]
Substituting into \eqref{eq:trunc_bound} gives
\begin{equation}\label{eq:D2_cond_lb}
\mathbb{E}[\Lambda_k\mid \bm X=\bm x]
\ge
a\,\big(1-F_{\bm x}(a)\big)^k.
\end{equation}

Averaging \eqref{eq:D2_cond_lb} over $\bm X$ yields
\[
D_2(k)
=
\mathbb{E}[\Lambda_k]
\ge
a\,\mathbb{E}_{\bm X}\!\left[\big(1-F_{\bm X}(a)\big)^k\right].
\]
Define $\phi(u)=(1-u)^k$ for $u\in[0,1]$. Since $\phi$ is convex on $[0,1]$ for every integer $k\ge 1$, Jensen's inequality gives
\[
\mathbb{E}_{\bm X}\!\left[\big(1-F_{\bm X}(a)\big)^k\right]
=
\mathbb{E}_{\bm X}[\phi(F_{\bm X}(a))]
\ge
\phi\!\left(\mathbb{E}_{\bm X}[F_{\bm X}(a)]\right).
\]
The key point is that only the \emph{averaged} small-ball probability appears:
\[
\mathbb{E}_{\bm X}[F_{\bm X}(a)]
=
\mathbb{P}(W\le a).
\]
Hence, under the averaged small-ball condition \eqref{eq:ball_assumption}, for any $a\in[0,a_0]$ satisfying $Ca^\alpha\le 1$,
\begin{equation}\label{eq:D2_lb_a}
D_2(k)
\ge
a\,\big(1-Ca^\alpha\big)^k.
\end{equation}

We now optimize the explicit lower bound in \eqref{eq:D2_lb_a}. Consider
\[
L(a)\triangleq a(1-Ca^\alpha)^k,
\qquad
0\le a\le \min\{a_0,C^{-1/\alpha}\}.
\]
A direct derivative calculation gives
\[
L'(a)
=
(1-Ca^\alpha)^{k-1}\Big(1-(1+\alpha k)Ca^\alpha\Big).
\]
Therefore, $L(a)$ is maximized at the stationary point
\[
a^\star
=
\left(\frac{1}{C(1+\alpha k)}\right)^{1/\alpha}.
\]
By assumption in the theorem statement, $a^\star\le a_0$, and clearly
\[
Ca^{\star\,\alpha}=\frac{1}{1+\alpha k}<1,
\]
so $a^\star$ is admissible in \eqref{eq:D2_lb_a}. Substituting $a=a^\star$ into \eqref{eq:D2_lb_a} yields
\[
D_2(k)
\ge
\left(\frac{1}{C(1+\alpha k)}\right)^{1/\alpha}
\left(1-\frac{1}{1+\alpha k}\right)^k.
\]
To bound the second factor, use
\[
\log(1-t)\ge -\frac{t}{1-t},
\qquad 0\le t<1,
\]
with $t=\frac{1}{1+\alpha k}$. Then
\[
\left(1-\frac{1}{1+\alpha k}\right)^k
=
\exp\!\left(
k\log\!\left(1-\frac{1}{1+\alpha k}\right)
\right)
\ge
\exp\!\left(
-\frac{k\cdot \frac{1}{1+\alpha k}}{1-\frac{1}{1+\alpha k}}
\right)
=
\e^{-1/\alpha}.
\]
Consequently,
\[
D_2(k)
\ge
\e^{-1/\alpha}
\left(\frac{1}{C(1+\alpha k)}\right)^{1/\alpha},
\]
which is exactly \eqref{eq:D2_small_ball_bound}. In particular, as $k\to\infty$,
\[
D_2(k)=\Omega\!\big(k^{-1/\alpha}\big).
\]

\section{Small-Ball Bounds under Local Power-Law Density Conditions}\label{app_2}

This appendix converts local power-law bounds on the conditional density of the MMSE error into small-ball probability bounds. In the main text, only the upper-bound direction is needed in
Case~2 of Section~\ref{subsec:d-dim-rate-comparison}; the lower-bound
direction is included to show how matching small-ball exponents
follow under two-sided density bounds.

\begin{lem}[Local power-law density bound implies small-ball bound]
\label{lem:powerlaw_smallball}
Assume that conditioned on $\bm X=\bm x$, the error $\bm Z$ admits a Lebesgue density $f_{\bm Z|\bm X}(\cdot|\bm x)$ on $\mathbb{R}^d$. Fix $\beta>-d$ and $r>0$.

Suppose first that there exists a nonnegative function $c_{\max}(\bm x)$ such that, for all $\bm x$ and all $\bm z$ with $\|\bm z\|_2\le r$,
\begin{equation}\label{eq:powerlaw_upper}
f_{\bm Z|\bm X}(\bm z|\bm x)
\le
c_{\max}(\bm x)\,\|\bm z\|_2^{\beta}.
\end{equation}
Then, for all $a\in[0,r^2]$ and all $\bm x$,
\begin{equation}\label{eq:smallball_powerlaw_upper}
\mathbb{P}\!\left(\|\bm Z\|_2^2\le a \,\middle|\,\bm X=\bm x\right)
\le
\frac{S_d}{d+\beta}\,c_{\max}(\bm x)\,a^{(d+\beta)/2},
\end{equation}
where $S_d=dV_d$ is the surface area of the unit sphere in $\mathbb{R}^d$.

Moreover, if there also exists a nonnegative function $c_{\min}(\bm x)$ such that, for all $\bm x$ and all $\bm z$ with $\|\bm z\|_2\le r$,
\begin{equation}\label{eq:powerlaw_lower}
c_{\min}(\bm x)\,\|\bm z\|_2^{\beta}
\le
f_{\bm Z|\bm X}(\bm z|\bm x),
\end{equation}
then for all $a\in[0,r^2]$ and all $\bm x$,
\begin{equation}\label{eq:smallball_powerlaw_twosided}
\frac{S_d}{d+\beta}\,c_{\min}(\bm x)\,a^{(d+\beta)/2}
\le
\mathbb{P}\!\left(\|\bm Z\|_2^2\le a \,\middle|\,\bm X=\bm x\right)
\le
\frac{S_d}{d+\beta}\,c_{\max}(\bm x)\,a^{(d+\beta)/2}.
\end{equation}
Consequently, whenever $0<c_{\min}(\bm x)\le c_{\max}(\bm x)<\infty$, the conditional small-ball probability is of order $a^{(d+\beta)/2}$ as $a\downarrow 0$.
\end{lem}
\begin{IEEEproof}
Fix $\bm x$ and $a\in[0,r^2]$. By polar coordinates,
\[
\mathbb{P}(\|\bm Z\|_2^2\le a\mid \bm X=\bm x)
=
\int_{\|\bm z\|_2\le \sqrt a} f_{\bm Z|\bm X}(\bm z|\bm x)\,\D\bm z
=
\int_0^{\sqrt a}\int_{\mathbb{S}^{d-1}}
f_{\bm Z|\bm X}(\rho \bm u|\bm x)\,\rho^{d-1}\,\D\bm u\,\D\rho,
\]
where $\mathbb{S}^{d-1}$ is the unit sphere and
\[
\int_{\mathbb{S}^{d-1}}\D\bm u=S_d=dV_d.
\]

If \eqref{eq:powerlaw_upper} holds, then for $\rho\le r$,
\[
f_{\bm Z|\bm X}(\rho\bm u|\bm x)\le c_{\max}(\bm x)\rho^\beta.
\]
Therefore,
\[
\mathbb{P}(\|\bm Z\|_2^2\le a\mid \bm X=\bm x)
\le
c_{\max}(\bm x)
\int_0^{\sqrt a}\int_{\mathbb{S}^{d-1}}
\rho^\beta \rho^{d-1}\,\D\bm u\,\D\rho
=
c_{\max}(\bm x)\,S_d
\int_0^{\sqrt a}\rho^{d+\beta-1}\,\D\rho.
\]
Since $\beta>-d$,
\[
\int_0^{\sqrt a}\rho^{d+\beta-1}\,\D\rho
=
\frac{(\sqrt a)^{d+\beta}}{d+\beta}
=
\frac{a^{(d+\beta)/2}}{d+\beta},
\]
which proves \eqref{eq:smallball_powerlaw_upper}.

If, in addition, \eqref{eq:powerlaw_lower} holds, then the same calculation with the lower bound gives
\[
\mathbb{P}(\|\bm Z\|_2^2\le a\mid \bm X=\bm x)
\ge
c_{\min}(\bm x)\,S_d
\int_0^{\sqrt a}\rho^{d+\beta-1}\,\D\rho
=
\frac{S_d}{d+\beta}\,c_{\min}(\bm x)\,a^{(d+\beta)/2},
\]
which completes the proof of \eqref{eq:smallball_powerlaw_twosided}.
\end{IEEEproof}

\begin{revision}
\subsection{Verification for the vanishing-density example}\label{app:vanishing_verification}

For the example in Section~\ref{subsec:vanishing_example}, the conditional density of the MMSE error satisfies
\[
f_{\bm Z|\bm X}(\bm z|\bm x)
=\frac{\varphi_d(\bm x-\bm z)\|\bm z\|_2^2\varphi_d(\bm z)}{d f_{\bm X}(\bm x)}
\le c_{\max}(\bm x)\|\bm z\|_2^2,\qquad \|\bm z\|_2\le1,
\]
where
\[
c_{\max}(\bm x)=\frac{\varphi_d(\bm 0)}{d f_{\bm X}(\bm x)}
\sup_{\|\bm u\|_2\le1}\varphi_d(\bm x-\bm u).
\]
Moreover,
\[
\mathbb E[c_{\max}(\bm X)]
=\frac{\varphi_d(\bm 0)}{d}
\int_{\mathbb R^d}\sup_{\|\bm u\|_2\le1}\varphi_d(\bm x-\bm u)\,\D\bm x<\infty,
\]
since the integrand is bounded on the unit ball and has Gaussian radial tails outside it. Hence the conditional vanishing-density condition holds with $\beta=2$ and $r=1$.
\end{revision}

\begin{revision}
\section{Other Decentralized Estimators}\label{app:other_estimators}

\subsection{Proof of Proposition~\ref{prop:general_local_bound}}\label{app:general_local_proof}

Conditioning on $(\bm Y_i,U_i)$ leaves the posterior density of
$\bm X$ equal to $f_{\bm X|\bm Y}(\cdot|\bm Y_i)$. For $a\ge0$,
\[
\begin{aligned}
&\mathbb P\!\left(
\|\bm X-g_i(\bm Y_i,U_i)\|_2^2\le a
\,\middle|\,\bm Y_i,U_i\right)\\
&\quad=
\int_{\|\bm x-g_i(\bm Y_i,U_i)\|_2\le\sqrt a}
f_{\bm X|\bm Y}(\bm x|\bm Y_i)\,\D\bm x\\
&\quad\le
V_d a^{d/2}
\operatorname*{ess\,sup}_{\bm x\in\mathbb R^d}
f_{\bm X|\bm Y}(\bm x|\bm Y_i).
\end{aligned}
\]
Averaging and using \eqref{eq:posterior_sup_mean} gives
\[
\mathbb{P}\!\left(\|\bm X-g_i(\bm Y_i,U_i)\|_2^2\le a\right)
\le V_d\overline M a^{d/2},\qquad a\ge0.
\]
For $\Lambda_k=\min_i\|\bm X-g_i(\bm Y_i,U_i)\|_2^2$, the union bound yields $\mathbb{P}(\Lambda_k>a)\ge1-kV_d\overline M a^{d/2}$. Thus
\[
\mathbb{E}[\Lambda_k]\ge
\int_0^{(kV_d\overline M)^{-2/d}}
\left(1-kV_d\overline M a^{d/2}\right)\,\D a
=\frac{d}{d+2}(kV_d\overline M)^{-2/d}.
\]
No independence of the estimation errors is needed.

For a common function with i.i.d.\ private seeds, Theorem~\ref{th_2} also applies directly to the augmented observations $(\bm Y_i,U_i)$ whenever the induced error satisfies \eqref{eq:ball_assumption}.

\subsection{Example: A two-agent comparison between the MMSE rule and a non-MMSE rule of the decentralized benchmark}\label{ex:non_mmse_rule}

Let $X$, $N_1$, and $N_2$ be independent $\mathcal{N}(0,1)$ random variables and $Y_i=X+N_i$. Both agents use $g_c(y)=cy$. The errors $Z_i=X-cY_i$ have variance $v=(1-c)^2+c^2$ and covariance $w=(1-c)^2$. Since $Z_1+Z_2$ and $Z_1-Z_2$ are independent centered Gaussians with variances $2(v+w)$ and $2(v-w)$,
\begin{align}
D_{2,g_c}(2)
&=\mathbb E\!\left[
\min\{(X-cY_1)^2,(X-cY_2)^2\}\right]\notag\\
&=\frac12\mathbb E\!\left[
Z_1^2+Z_2^2-|Z_1^2-Z_2^2|\right]\notag\\
&=\frac12\mathbb E\!\left[
Z_1^2+Z_2^2-|Z_1+Z_2|\,|Z_1-Z_2|\right]\notag\\
&=\frac12\mathbb E\!\left[
Z_1^2 \right] + E\!\left[Z_2^2\right] - \frac12 \mathbb E[|Z_1+Z_2|]\,\mathbb E[|Z_1-Z_2|]\notag\\
&=v-\frac{2}{\pi}\sqrt{v^2-w^2},
\label{eq:two_agent_comparison}
\end{align}
where the last equality uses the identity
$\mathbb E[|G|]=\sigma\sqrt{2/\pi}$ for
$G\sim\mathcal N(0,\sigma^2)$ and the independence of
$Z_1+Z_2$ and $Z_1-Z_2$, which give
$
\mathbb E[|Z_1+Z_2|]\,\mathbb E[|Z_1-Z_2|] =\frac{4}{\pi}\sqrt{v^2-w^2}.
$

The MMSE rule corresponds to $c=1/2$, whereas $c=3/5$ gives
\begin{align*}
D_2(2)&=\frac12-\frac{\sqrt3}{2\pi}\simeq0.22434,\\
D_{2,g_{3/5}}(2)&=\frac{13}{25}-\frac{6\sqrt{17}}{25\pi}
\simeq0.20502<D_2(2).
\end{align*}
The single-agent MSE increases from $1/2$ to $13/25$. Thus a common deterministic non-MMSE rule can improve the minimum-over-agents distortion despite having worse individual MSE. This example establishes an improvement over the MMSE benchmark for $k=2$. However, the uniform lower bound in
Section~\ref{subsec:other_estimators} establishes that no admissible local rule of the decentralized benchmark can achieve a faster decay rate than the $k^{-2/d}$ rate of the centralized list estimator.
\end{revision}

\bibliographystyle{IEEEtran}
\bibliography{Citations}

\end{document}